\documentclass[letterpaper, 10 pt, conference]{ieeeconf}  

\IEEEoverridecommandlockouts                              

\usepackage{graphicx}      
\usepackage{amsmath}
\usepackage{amssymb}
\usepackage{amsfonts}
\usepackage{booktabs}
\usepackage{graphicx}
\usepackage{tikz}
\usetikzlibrary{arrows.meta}
\usetikzlibrary{positioning}
\usepackage{lineno}
\usepackage{bm}

\title{\LARGE \bf
Optimal Steady-State Regulation of Nonlinear Systems: A Feedback Optimization Approach}

\author{Wei He, Lei Wang, Jingjing Fang, Zheng-Guang Wu, and Hongye Su
\thanks{*This work was supported by Zhejiang Provincial Natural Science Foundation of China under Grant No. LR26F030003 and Natural Science Foundation of China under Grant No. 62573383.}
\thanks{W. He, L. Wang, Z. Wu, and H. Su are with Institute of Cyber-Systems and Control, Zhejiang University, Hangzhou, 310027, China. J. Fang is  with Zhejiang Development \& Planning Institute.
         Corresponding author: Lei Wang ({\tt\small lei.wangzju@zju.edu.cn})}%
    }
\newtheorem{problem}{Problem}
\newtheorem{assumption}{Assumption}
\newtheorem{theorem}{Theorem}
\newtheorem{remark}{Remark}
\newtheorem{lemma}{Lemma}

\begin{document}

\maketitle
\thispagestyle{empty}
\pagestyle{empty}

\begin{abstract}

This paper studies the optimal steady-state regulation problem for a class of nonlinear systems. The regulated output is required to track a trajectory that minimizes a steady-state cost depending on both the output and the exo-system. To characterize the desired steady-state behavior, we introduce a generalized regulator equation and develop a two-layer framework. The optimization layer generates the optimal reference online through a steady-state generator combined with gradient-feedback optimization dynamics, while the control layer employs backstepping to stabilize the plant and achieve the desired optimal steady-state regulation. Under suitable immersion and minimum-phase assumptions, global asymptotic stability is established. 

\end{abstract}

\section{Introduction}\label{sec:introduction}
The problem of output regulation investigates a plant cascaded with an exo-system, and aims at designing a feedback regulator driven by regulation output and extra measurement to steer the regulation output to a given point. The internal model-based method was first introduced and shown to be an effective design and analysis tool in \cite{francis1976internal}, and has been well developed for nonlinear systems \cite{isidori1990output}, leading to the Regulator Equation, and for hybrid systems \cite{marconi2013internal} and networked systems \cite{wang2019robust}. To design the internal model, \cite{huang2004nonlinear} introduced the immersion condition, which is removed by \cite{marconi2008uniform}, establishing the existence of regulator in the pre-processing and minimum-phase setting. In \cite{wang2020pre} and \cite{wang2022robust}, which considered systems with extra measurement, a washout filter is constructed based on the internal model, blocking the unknown non-vanishing measurement in steady state.

However, in some cases, the regulated output is required to track an optimal steady-state trajectory that is not prescribed a priori but instead determined by an optimization problem. In this case, the reference signal often needs to be determined by the system itself. \cite{jokic2009constrained} defined the optimal steady-state regulation (OSSR) problem for a class of disturbed nonlinear systems with equality and inequality constraints, and designed control algorithms based on barrier functions and Karush-Kuhn-Tucker conditions, respectively.  \cite{lawrence2020linear} provided a controller design framework for the convex optimal steady-state control problem of disturbed linear systems, and proved that the existence of the proposed optimality model is a sufficient condition for reducing the optimal steady-state control problem to a stabilization problem. Based on this, \cite{simpson2022low} proposed several design methods for optimality models and stabilizers based on low-gain integral control and timescale separation analysis.
In the practical distributed power energy systems, Li et al. studied the problem from the perspective of optimal steady-state control \cite{li2015connecting,zhang2017distributed}. In \cite{li2015connecting}, the relationship between Automatic Generation Control and Economic Dispatch Control was investigated from an optimization viewpoint, combining them to propose a distributed optimization algorithm to improve the economic efficiency of classical control strategy. \cite{zhang2017distributed} proposed a general control update framework and studied the necessary and sufficient conditions under which a linear time-invariant system can be interpreted as the gradient algorithm for unconstrained convex optimization or unconstrained saddle-point problems. Due to the widespread attention on distributed control in recent years, some research has introduced the OSSR problem into networked systems, proposing a class of distributed optimal steady-state control problems and achieving certain theoretical advances \cite{notarnicola2023gradient,jiang2024distributed,xian2025fully}.


In this paper, we consider the convex OSSR problem for a class of nonlinear systems in lower-triangular form, where the regulated output is expected to track a reference signal which minimizes the cost function involving the system output and a signal generated by an exo-system, with the gradient measurable.
The main contribution of this paper is that we construct a two-layer framework to solve the OSSR problem. By combining the regulator equation with the optimality condition, we formulate the Generalized Regulator Equation, based on which we design the dynamics of the optimization layer with the immersion condition. In the control layer, we construct the stabilizer based on backstepping to facilitate an ISS analysis of the system. Under the proposed framework, the global asymptotic stability of the optimal steady-state manifold is guaranteed.

The remainder of the paper is organized as follows. In Section \ref{sec:problem formulation} the problem is formulated and we present some preliminaries and standing assumptions. Section \ref{sec:main resuls} is the main result of this work, which proposes a two-layer framework and the methodology of designing the optimization and control layer respectively. Then, in Section \ref{sec:simulation} simulation examples are shown to illustrate the theoretical results. A brief conclusion is made in Section \ref{sec:conclusion}.

\section{Problem Formulation}
\label{sec:problem formulation}

Consider nonlinear systems in a lower triangular form as

\begin{equation}\label{nonlinear plant}
	\begin{cases}
		\dot{w}&=s(w) \\
		\dot{x}_0&=f_0(x_0,x_1,w)\\
		\dot{x}_1&=f_1(x_0,x_1,w)+b_1(w)x_2\\
		\dot{x}_2&=f_2(x_0,x_1,x_2,w)+b_2(w)x_3\\
	\end{cases}
\end{equation}
where  $w\in\mathbb{R}^{n_w}$ is the exogenous state, $x_0\in\mathbb{R}^{n_0}$ is the state of internal dynamics,   $y_r:=x_1\in\mathbb{R}$ is the regulated output, and $u:=x_3\in\mathbb{R}$ is the control input to be designed. 
Assume  $b_1(w),b_2(w)\geq b_{min}>0$ for all $w\in\mathcal{W}$.

As customary in the problem of output regulation, the $w$-exo-system is assumed to be neutrally stable, that is, there exists a compact set $\mathcal{W}\subset\mathbb{R}^{n_w}$ such that $w\in\mathcal{W}$ is invariant for the dynamics $\dot{w}=s(w)$. Note that in the conventional output regulation problem, the controller is designed to regulate the $y_r$ to zero or a given reference.
In this paper, we consider an optimal steady-state regulation (OSSR) problem (as given below), where the steady states of $y_r$ are unavailable, but determined by the plant \eqref{nonlinear plant} together with an optimization problem.

\begin{problem}\label{prob2}
	This paper aims to design a dynamic controller for the plant \eqref{nonlinear plant}, driven by $y_r$, such that in the presence of the exogenous signal $w(t)\in\mathcal{W}$, the resulting closed-loop trajectories are bounded, and the regulated output $y_r$ is asymptotically steered to a reference signal $y_r^\ast$, which minimizes the cost function $F_c({y}_r,w)$ in steady state.   
\end{problem}

\begin{remark}
	Though the results in this paper concern systems with relative degree $r=2$, the subsequent control design can be directly extended to systems with higher relative degrees.
\end{remark}


Next, we list some useful standing assumptions.

\begin{assumption}\label{A:convexity and smoothness}
	The cost function $F_c({y}_r,w)$ is sufficiently smooth, and  strongly convex with respect to ${y}_r$, uniformly in $w\in  \mathcal{W}$, in the sense that there exists $\omega_{\text{min}}>0$ such that
	\begin{equation*}
		\begin{aligned}
			({\nabla_{{y}_r} F_c({y}_1,w)-\nabla_{{y}_r} F_c({y}_2,w)})/({y_1-y_2}) \geq \omega_{\text{min}}
		\end{aligned}
	\end{equation*}
	holds for any distinct $y_1, y_2\in \mathbb{R}$ and all $w\in \mathcal{W}$. Here $\nabla_{{y}_r} F_c({y}_r,w)$ denotes the gradient with respect to $y_r$, and is \emph{measurable} for controller design.
\end{assumption}
\begin{remark}
	For simplicity, Assumption \ref{A:convexity and smoothness} restricts the cost function $F_c(\cdot,w)$ to be strongly convex. However, it is noted that due to the presence of unknown $w$, the optimum is not solvable directly. In this setting, its gradient is assumed to be measurable online for implementation of feedback optimization. Such assumption is feasible in many scenarios, e.g., the multiple-source localization problem with acoustic energy measurements available \cite{meng2020multiple} \cite{zhang2015distributed}, and the distributed optimal coordination problem \cite{zhang2017distributed}. 
\end{remark}

By adapting the classical nonlinear regulator equation to our OSSR problem setting, we make the following assumption about the existence of solutions for the Generalized Regulator Equation (GRE).
\begin{assumption}\label{A:GRE}
	There exist smooth functions $\pi_{0}:\mathcal{W}\to\mathbb{R}^{n_0}$, and $\pi_{1}:\mathcal{W}\to\mathbb{R}$ such that
	\begin{equation*} \label{eq:GRE}
		\begin{aligned}
			&\nabla F_c(\pi_{1}(w),w) =0 \\
			&\mathcal{L}_{s(w)}\pi_{0}(w)=f_0(\pi_{0}(w),\pi_{1}(w),w) \,.
		\end{aligned}
	\end{equation*}
\end{assumption}

Define error states $\tilde{x}_0:=x_0-\pi_{0}(w)$ and $\tilde{x}_1:=x_1-\pi_{1}(w)$. Then there follows the error dynamics of the form
\begin{equation}\label{eq:t_x_0}
	\begin{aligned}
		&\dot{\tilde{x}}_0=\tilde{f}_0(\tilde{x}_0,\tilde{x}_1,w) \\
	\end{aligned}
\end{equation}
where $\tilde{f}_0(\tilde{x}_0,\tilde{x}_1,w)={f}_0({x}_0,{x}_1,w)-{f}_0(\pi_{0}(w),\pi_{1}(w),w)$. Next, we impose the following strongly minimum-phase assumption.
\begin{assumption}\label{strong minimum phase property}
	The error dynamics \eqref{eq:t_x_0} is input-to-state stable (ISS) with state $\tilde x_0$ and input $\tilde x_1$, uniformly in $w\in\mathcal{W}$. Namely, there exists $C^1$ function $ V_0(\tilde{x}_0)$  such that 
	\begin{equation}\label{ineq:minimum phase}
		\begin{aligned}
			\underline{\alpha}_0(|\tilde{x}_0|)\leq V_0(\tilde{x}_0)\leq\bar{\alpha}_0(|\tilde{x}_0|)\\
			\dot{V}_0(\tilde{x}_0)\leq-\alpha_0(|\tilde{x}_0|)+\alpha_{r_1}(|\tilde{x}_1|)
		\end{aligned}
	\end{equation}
	where $\underline{\alpha}_0, \bar{\alpha}_0, \alpha_0, \alpha_{r_1}\in\mathcal{K}_{\infty}$.
\end{assumption}

Let \(\pi_{2}(w)=(\mathcal{L}_{s(w)}\pi_1(w)-f_1(\pi_0(w),\pi_1(w),w))/b_1(w)\), \(\pi_{3}(w)=(\mathcal{L}_{s(w)}\pi_2(w)-f_2(\pi_0(w),\pi_1(w),\pi_2(w),w))/b_2(w)\).
The next assumption, which is inherited from \cite{byrnes2004nonlinear} and restricts the maps $\pi_{i}(w)$, $i=1,2,3$ to fulfill the immersion conditions, is stated as follows.
\begin{assumption}\label{immersion condition assumption}
	There exist  positive numbers $d_i$ and bounded smooth  functions $\phi_{i}:\mathbb{R}^{d_i}\to\mathbb{R}$, $i=1,2,3$ such that for all $w\in\mathcal{W}$, 
	\begin{equation}\label{immersion condition}
		\mathcal{L}_{s(w)}^{d_i}\pi_i(w)=\phi_{i}(\pi_i(w),\mathcal{L}_{s(w)}\pi_i(w),...,\mathcal{L}_{s(w)}^{d_i-1}\pi_i(w)). \\[1mm]
	\end{equation}
\end{assumption}

\begin{remark}
    Assumption \ref{strong minimum phase property} is imposed to simplify the analysis, which could be relaxed by incorporating the extra measurement to stabilize the unstable internal dynamics as in \cite{wang2020pre}. Immersion conditions \eqref{immersion condition} are introduced to derive constructive design solutions and simplify the forthcoming discussion, though they are shown to be unnecessary in \cite{marconi2007output}. 
\end{remark}
\section{Main Results}\label{sec:main resuls}

The proposed controller design framework is composed of two layers: the optimization layer and the control layer.

\subsection{Design of optimization layer}

We propose the gradient measurement-based optimization dynamics as
\begin{equation}\label{eq:xi}
	\begin{cases}
		\dot{\xi}_r=A_r\xi_r+B_r\phi_{1}(\xi_r)-G_r\nabla F_c(y_r,w)+\kappa(y_r-y_r^r)\\
		y^r_r=C_r\xi_r
	\end{cases}
\end{equation}
where $\xi_r\in\mathbb{R}^{d_1}$, 
$(A_{r},B_{r},C_{r})\in\mathbb{R}^{d_1\times d_1}\times \mathbb{R}^{d_1}\times\mathbb{R}^{1\times d_1}$ is a triplet in the prime form, 
$\kappa(\cdot):\mathbb{R}\to\mathbb{R}^{d_1}$ is a smooth function, and 
\(G_r=\mbox{col}(g_r\lambda_1,...,g_r^{d_1}\lambda_{d_1}) \in\mathbb{R}^{d_1}\),

with $\lambda_i$ to be chosen and constant $g_r>0$.

\begin{remark}
   The dynamics in \eqref{eq:xi} consist of three parts. The term $A_r\xi_r+B_r\phi_1(\xi_r)$, motivated by the immersion condition, serves to reproduce the steady-state reference of $y_r$. The gradient-feedback term is used to enforce the optimality condition and hence steer the generated reference toward the minimizer of $F_c$. The correction term $\kappa(y_r-y_r^r)$ is further introduced to provide extra stability margin and to guarantee the desired closed-loop stability property.
\end{remark}

Then define $e=y_r-y_r^r$,
\(
		\pi_{\xi_r}(w)=\text{col}(\pi_{1}(w),\mathcal{L}_{s(w)}\pi_{1}(w),...,\mathcal{L}_{s(w)}^{d_{1}-1}\pi_{1}(w))
\),

 and $\tilde{\xi}_r=\xi_r-\pi_{\xi_r}(w)$.
With Assumption \ref{A:convexity and smoothness}, we can obtain from  \eqref{eq:xi} 
\begin{equation} \label{eq:txir}
	\begin{aligned}
		\dot{\tilde{\xi}}_r
		=&A_r\tilde{\xi}_r-G_r\Omega_1(C_r\tilde{\xi}_r,w)+B_r\tilde{\phi}_{1}(w,\tilde{\xi}_r)\\&-G_r \Omega_2(e,w)+\kappa(e)\\
	\end{aligned}
\end{equation}
where $\tilde{\phi}_{1}(\tilde{\xi}_r,w)={\phi}_{1}(\xi_r)-{\phi}_{1}(\pi_{\xi_r}(w))$, 
$\Omega_1(C_r\tilde{\xi}_r,w)=\nabla F_c(C_r\xi_r,w)-\nabla F_c(C_r\pi_{\xi_r}(w),w)$, $\Omega_2(e,w)=\nabla F_c(y_r, w)-\nabla F_c(C_r \xi_r,w)$.
It is clear that $\Omega_1(0,w)=0$ and $\Omega_2(0,w)=0$. Moreover,  by the Mean Value Theorem, we have $\Omega_1(C_r\tilde{\xi}_r,w)= \omega(\bar y_r^r,w)C_r\tilde{\xi}_r$ for some $\bar y_r^r$ taking value between $C_r\tilde{\xi}_r$ and $0$. For convenience, we omit the argument of function $\omega$ and note that by Assumption \ref{A:convexity and smoothness},  \(\omega \geq \omega_{\text{min}}>0\).

We first state a useful Lemma about the uniform Hurwitz property of a class of time-varying matrices.
\begin{lemma} \label{Lemma Hurwitz}
	Given a continuous signal $\omega(t)$ with
	$\omega(t)\ge\omega_{\min}>0$ and a matrix $M_n(t)\in\mathbb{R}^{n\times n}$
	of the form
	\[
	M_n(t)=
	\begin{pmatrix}
		-a_{n-1} \omega(t) & 1 & 0 & \cdots & 0\\
		-a_{n-2} \omega(t) & 0 & 1 & \cdots & 0\\
		\vdots               & \vdots & \ddots & \ddots & \vdots\\
		-a_{1}\omega(t) & 0 & \cdots & 0 & 1\\
		-a_0 \omega(t) & 0 & \cdots & 0 & 0
	\end{pmatrix},
	\]
	there exist real numbers $a_0,\dots,a_{n-1}$, a constant
	$\lambda_n>0$ and a symmetric positive definite matrix
	$P_n\in\mathbb{R}^{n\times n}$
	such that
	\begin{equation*}\label{eq:An}
		M_n(t)^\top P_n + P_n M_n(t) \le -\lambda_n I_n,\qquad \forall t\ge0.
	\end{equation*}
	
\end{lemma}
The proof is directly adapted from Lemma 2.1 in \cite{gauthier2001deterministic} and is thus omitted.


Then we state the ISS property of the dynamics in the optimization layer, with the proof given in the Appendix A.
\begin{lemma}\label{lemma 1}
	Under Assumptions \ref{A:convexity and smoothness}, \ref{A:GRE}  and \ref{immersion condition assumption},  there  exists $g_r^*>0$ such that for any $g_r>g_r^*$ and smooth function $\kappa(\cdot)$, the $\tilde{\xi}_r$-subsystem \eqref{eq:txir} is ISS with state $\tilde{\xi}_r$ and input $e$. More specifically, there exists a symmetric positive definite matrix $Q$ such that 
	\begin{equation*}
		\dot{V}_o(\tilde{\xi}_r)
		\leq -\alpha_{\xi}|\tilde{\xi}_r|^2+\alpha_{e}(|e|)
	\end{equation*}
	with $V_o(\tilde{\xi}_r)=\tilde{\xi}_r^\top Q \tilde{\xi}_r$, $\alpha_{\xi}>0$, and $\alpha_{e}(\cdot)\in\mathcal{K}_\infty$.
\end{lemma}

\subsection{Design of control layer}

In this subsection, we proceed to complete the design of the control layer, which, together with the previously designed optimization layer \eqref{eq:xi}, solves the desired optimal steady-state regulation problem. 

The design of the control layer follows the idea of backstepping \cite{krstic1995nonlinear}. We first consider the $x_1$ subsystem 
by regarding $x_2$ as the control input, and design
\begin{equation}\label{internal model 1}
	\dot{\eta}_1=A_1\eta_1+B_1\phi_{2}(\eta_1)+G_1v_1
\end{equation}
where $\eta_1\in\mathbb{R}^{d_2}$, with $(A_{1},B_{1},C_{1})\in\mathbb{R}^{d_2\times d_2}\times \mathbb{R}^{d_2}\times\mathbb{R}^{1\times d_2}$ a triplet in prime form, $G_{1}=\operatorname{col}(g_1\nu_1,...,g_1^{d_2}\nu_{d_2}) $
with $\nu_i$ being coefficients of a Hurwitz polynomial, and $v_1$ is determined later.

Next we study the ISS property of the $(\xi_r,x_0,e,\eta_1)$-subsystem.
Define vector $H_1=(\nu_1,...,\nu_{d_2})^\top$.
Define $\tilde{\eta}_1:=\eta_1-\pi_{\eta_1}(w)$, where
\begin{equation*}\label{eta1 ss}
	\pi_{\eta_1}(w)=\text{col}(\pi_{2}(w),\mathcal{L}_{s(w)}\pi_{2}(w),...,\mathcal{L}_{s(w)}^{d_{2}-1}\pi_{2}(w)) 
\end{equation*}
and we can obtain from \eqref{internal model 1} and \eqref{immersion condition} that
\begin{equation*}\label{eq:error eta1}
	\dot{\tilde{\eta}}_1=A_1\tilde{\eta}_1+B_1\tilde{\phi}_{2}(\tilde{\eta}_1,w)+G_1v_1,
\end{equation*}
where $\tilde{\phi}_{2}(\tilde{\eta}_1,w)=\phi_{2}(\eta_1)-\phi_{2}(\pi_{\eta_1}(w))$. 
The dynamics of $\tilde{x}_1$ are as follows
\begin{equation*}
	\begin{aligned}
		\dot{\tilde{x}}_1=&\tilde{f}_1(\tilde{x}_0,\tilde{x}_1,w)+b_1(w)C_1\tilde{\eta}_1+b_1(w)D_1v_1\\&+b_1(w)(x_2-C_1\eta_1-D_1v_1)
	\end{aligned}
\end{equation*}
where $\tilde{f}_1(\tilde{x}_0,\tilde{x}_1,w)=f_1(x_0,x_1,w)-f_1(\pi_{0}(w),\pi_{1}(w),w)$. 

We define $\vartheta_1=x_1-y_r^r$, which gives $\vartheta_1=\tilde{x}_1-C_r\tilde\xi_r$.
By defining $\bar{\eta}_1=\tilde{\eta}_1-b_1^{-1}(w)G_1D_1^{-1}\vartheta_1$, $\vartheta_2=x_2-C_1\eta_1-D_1v_1$, we have
\begin{equation*}
	\begin{aligned}		\dot{\vartheta}_1=&\tilde{f}_1(\tilde{x}_0,\vartheta_1+C_r\tilde\xi_r,w)+b_1(w)C_1[\bar{\eta}_1+b_1^{-1}(w)G_1D_1^{-1}\vartheta_1]\\&+b_1(w)\vartheta_2-C_r\dot {\tilde \xi}_r+b_1(w)D_1v_1
	\end{aligned}
\end{equation*}

Define $\tilde z_1=\text{col}(\tilde \xi_r,\tilde x_0,\bar \eta_1)$,  $\bar{z}_1=\text{col}(\tilde{z}_1,\vartheta_1)$ with dynamics $\dot{\bar{z}}_1=\bar{f}_1(\bar{z}_1,\vartheta_2,w)$, where the smooth function $\bar f_1(\bar z_1,\vartheta_2,w)$ satisfies $\bar f_1(0,0,w)=0$. The following lemma characterizes the ISS property of the $\bar z_1$-subsystem, with the proof given in the Appendix B.
\begin{lemma}\label{lemma 2}
	Under Assumptions \ref{A:convexity and smoothness}-\ref{immersion condition assumption}, there  exist $g_1^*>0$, and a smooth control law $v_1=-\rho_1(\vartheta_1)$, 
	such that for any $g_1>g_1^*$ the $(\tilde{\xi}_r,\tilde x_0,\bar{\eta}_1,\vartheta_1)$-subsystem is ISS, with $(\tilde{\xi}_r,\tilde x_0,\bar{\eta}_1,\vartheta_1)$ as state and $\vartheta_2$ as input. Moreover, there exists an ISS Lyapunov function $V_{c_1}(\bar{z}_1)$, which satisfies
	\begin{equation*}
		\begin{aligned}
			\underline{\alpha}_{z_1}(|\bar{z}_1|^2)\leq V_{c_1}(\bar{z}_1)\leq\bar{\alpha}_{z_1}(|\bar{z}_1|^2)\\
			\dot{V}_{c_1}(\bar{z}_1)\leq -{\alpha}_{z_1}(|\bar{z}_1|^2)+{\alpha}_{\vartheta_2}(|\vartheta_2|) \\
		\end{aligned}
	\end{equation*}
	where ${\alpha}_{z_1}, {\alpha}_{\vartheta_2}\in\mathcal{K}_{\infty}$.
\end{lemma}

With Lemma \ref{lemma 2} in mind,  we proceed to consider the $x_2$-subsystem and construct the pre-processing internal model as
\begin{equation}\label{internal model 2}
	\begin{cases}
		\dot{\eta}_2=A_2\eta_2+B_2\phi_{3}(\eta_2)+G_2v_2\\
		u=C_2\eta_2+D_2v_2
	\end{cases}
\end{equation}
where $\eta_2\in\mathbb{R}^{d_3}$, with $(A_{2},B_{2},C_{2})\in\mathbb{R}^{d_3\times d_3}\times \mathbb{R}^{d_3}\times\mathbb{R}^{1\times d_3}$ a triplet in prime form, $v_2$ is the residual control input to be designed, and 
\begin{equation}
	G_2=\operatorname{col}(g_2\mu_1,...,g_2^{d_3}\mu_{d_3})
\end{equation}
with $\mu_i$ being coefficients of Hurwitz polynomials, and $g_2$ is determined later. Define the vector $H_2=\text{col}(\mu_1,...,\mu_{d_3})$.

Define $\tilde{x}_2=x_2-\pi_{2}(w)$, $\tilde{\eta}_2=\eta_2-\pi_{\eta_2}(w)$, where
\begin{equation*}\label{steady state of eta1 and eta2}
	\begin{aligned}
		\pi_{\eta_2}(w)&=\text{col}(\pi_{3}(w),\mathcal{L}_{s(w)}\pi_{3}(w),...,\mathcal{L}_{s(w)}^{d_{3}-1}\pi_{3}(w)) 
	\end{aligned}
\end{equation*}
Defining $\bar{\eta}_2=\tilde{\eta}_2-b_2^{-1}(w)G_2D_2^{-1}\vartheta_2$, $\tilde{z}_2=\text{col}(\bar{z}_1,\bar{\eta}_2)$,
the extended error system is modeled as
\begin{equation} \label{steady state of z and x error form}
	\begin{aligned}
		\dot{\bar{z}}_1=&\bar{f}_1(\bar{z}_1,\vartheta_2,w) \\
		\dot{\bar{\eta}}_2=&M_{\eta_2}\bar{\eta}_2+B_2\tilde{\phi}_3(\bar{\eta}_2+b_2^{-1}(w)G_2D_2^{-1}\vartheta_2,w)\\&+G_2\Delta_{\eta_2}(\bar{z}_1,\vartheta_2,w)\\
		\dot{\vartheta}_2=&\varpi_{\vartheta_2}(\bar{z}_1,\bar{\eta}_2,\vartheta_2,w) 	+b_2(w)D_2v_2
	\end{aligned}
\end{equation}
where $M_{\eta_2}=A_2-G_2D_2^{-1}C_2$, $\tilde{\phi}_3(\tilde{\eta}_2,w)=\phi_{3}(\eta_2)-\phi_{3}(\pi_{\eta_2}(w))$ and
\begin{equation*} 
	\begin{aligned}
		&\varpi_{\vartheta_2}(\bar{z}_1,\bar{\eta}_2,\vartheta_2,w)=\\&\tilde{f}_2(\bar{z}_1,\vartheta_2,w)+b_2(w)C_2(\bar\eta_2+b_2^{-1}(w)G_2D_2^{-1}\vartheta_2)\\&-C_1\dfrac{d}{dt}(\bar\eta_1+b_1^{-1}(w)G_1D_1^{-1}\vartheta_1)-D_1\dfrac{d\rho_1(\vartheta_1)}{d\vartheta_1}\dot{\vartheta}_1\\
		&\tilde{f}_2(\bar{z}_1,\vartheta_2,w)=\\&f_2(x_0,x_1,x_2,w)-f_2(\pi_{0}(w),\pi_{1}(w),\pi_{2}(w),w) \\
	\end{aligned}
\end{equation*}
The smooth function $\Delta_{\eta_2}(\bar{z}_1,\vartheta_2,w)$ and $\varpi_{\vartheta_2}(\bar{z}_1,\bar\eta_2,\vartheta_2,w)$ satisfy $\Delta_{\eta_2}(0,0,w)=0$ and $\varpi_{\vartheta_2}(0,0,0,w)=0$.
There exist positive constants $L_{\phi_3}$, $L_{\phi_3}^{'}$, class $\mathcal{K}_\infty$ functions $\varphi_{z_1}(\cdot)$,  $\varphi_{\vartheta_2}(\cdot)$ such that
\begin{equation*}
	\begin{aligned}
		&|\tilde{\phi}_3(\bar{\eta}_2+b_2^{-1}(w)G_2D_2^{-1}\vartheta_2,w)|\leq L_{\phi_3}|\bar \eta_2|+ |G_2|L_{\phi_3}^{'}|\vartheta_2|\\
		&|{\Delta_{\eta_2}}|^2\leq \varphi_{z_1}(|\bar z_1|)+\varphi_{\vartheta_2}(|\vartheta_2|)\\
	\end{aligned}
\end{equation*}
and we let $\bar\varphi_{\vartheta_2}(|\vartheta_2|)=L_{\phi_3}^{'2}|\vartheta_2|^2+\varphi_{\vartheta_2}(|\vartheta_2|)$, which is a $\mathcal{K}_\infty$ function.
Moreover, there exist positive constant $\tau_{\eta_2}$ and class $\mathcal{K}_\infty$ functions $\tau_{\vartheta_2}(\cdot)$, $\tau_{{z}_1}(\cdot)$ such that
\begin{equation*}
	\begin{aligned}
		\vartheta_2\varpi_{\vartheta_2}(\bar{z}_1,\bar{\eta}_2,\vartheta_2,w)&\leq \tau_{z_1}(|\bar z_1|)+\tau_{\eta_2}|\bar{\eta}_2|^2+\tau_{\vartheta_2}(|\vartheta_2|)\\
	\end{aligned}
\end{equation*}

It turns out that there exist a matrix $G_2$ and a choice of $v_2 $ of the same kind of $v_1$, such that the origin of the closed-loop system \eqref{steady state of z and x error form} is globally asymptotically stable. This is detailed in the next theorem.

\begin{theorem}
	Under Assumptions \ref{A:convexity and smoothness}-\ref{immersion condition assumption}, let $g_r$ satisfy Lemma \ref{lemma 1}, and take $G_1$, $v_1$ as in Lemma \ref{lemma 2}. Then there exist $g_2^*>0$, and a smooth control law $v_2=-\rho_2(\vartheta_2)$,
	such that for any $g_2>g_2^*$, the closed-loop system composed of \eqref{nonlinear plant}, \eqref{eq:xi}, \eqref{internal model 1} and \eqref{internal model 2} globally asymptotically converges to the optimal steady-state manifold $y_r=\pi_{1}(w)$, and hence solves the OSSR problem.
\end{theorem}
\begin{proof}
	We consider the $(\bar z_1,\bar\eta_2)$-subsystem now, which is treated as zero dynamics, while the $\vartheta_2$-subsystem is regarded as the external subsystem of the closed-loop system $(\bar{z}_1,\bar\eta_2,\vartheta_2)$. Observing that the $(\bar z_1,\bar\eta_2)$-subsystem is a cascaded system, we only need to study the stability of the interconnected subsystems $(\bar z_1,\vartheta_2)$ and $(\bar\eta_2,\vartheta_2)$.
	
	Defining $M_{\eta_2}^{'}=A_2-H_2D_2^{-1}C_1$, there exists a positive definite matrix $P_{\eta_2}$ satisfying
	\begin{equation*} 
		\begin{aligned}
			P_{\eta_2}M_{\eta_2}^{'}+M^{'\top}_{\eta_2}P_{\eta_2}\leq -\mu_0I 
		\end{aligned}
	\end{equation*}
	for some $\mu_0>0$. Define $D_{g_2}=\text{diag}(g_2,...,g_2^{d_3})$.
	By the same argument as the proof of Lemma \ref{lemma 2}, by reconstructing $V_{\vartheta_2}(\vartheta_2)=1/2\vartheta_2^2$ and  $V_{\bar\eta_2}(\bar{\eta}_2)=\bar{\eta}_2^\top Q_{\eta_2}\bar{\eta}_2$, where $Q_{\eta_2}=(D_{g_2}^{-1})^\top P_{\eta_2}D_{g_2}^{-1}$, we have that 
	\begin{equation*}
		\begin{aligned}
			\dot{V}_{\bar \eta_2}(\bar{\eta}_2)\leq& -\dfrac{g_2\mu_0-2|H_2|^2| P_{\eta_2}|^2-2L_{\phi_3}|P_{\eta_2}|}{\lambda_{max}(P_{\eta_2})}V_{\bar \eta_2}(\bar{\eta}_2)\\
			& +\varphi_{z_1}(|\bar{z}_1|)+\bar\varphi_{\vartheta_2}(\sqrt{2V_{\vartheta_2}(\vartheta_2)})
		\end{aligned}
	\end{equation*}
	Namely, there exists a $g_2^*$ such that for any $g_2>g_2^*$, the $\bar{\eta}_2$ dynamics is ISS with respect to the input $(\bar{z}_1,\vartheta_2)$.
	With the feedback control law designed as $v_2=-\rho_2(\vartheta_2)$, we have
	\begin{equation*}
		\begin{aligned}
			\dot{V}_{\vartheta_2}(\vartheta_2)=&\vartheta_2(\varpi_{\vartheta_2}(\bar z_1,\bar\eta_2,\vartheta_2,w)-b_2(w)D_2 \rho_2(\vartheta_2))\\
			\leq&  -(b_{min}D_2 \rho_2(\vartheta_2)\vartheta_2-\tau_{\vartheta_2}(|\vartheta_2|))\\
			&+ \tau_{z_1}(\underline{\alpha}_{z_1}^{-1}(V_{c_1}(\bar z_1)))
			+\tau_{\eta_2}{{V_{\bar\eta_2}(\bar\eta_2)}/{\lambda_{min}(Q_{\eta_2})}}\\
		\end{aligned}
	\end{equation*}
	To guarantee the stability of the interconnected systems, with any $g_2\geq g_2^*$, we design the smooth function $\rho_2(\cdot)$ to satisfy the following
	\begin{equation}\label{ineq:smallgain rho2}
		\begin{aligned}
			\rho_2(\vartheta_2)\geq (\tau_{\vartheta_2}(\vartheta_2)+\psi_2(1/2\vartheta_2^2))/({b_{min}|D_2|\vartheta_2})
		\end{aligned}
	\end{equation}
	where $\psi_2(\cdot)>\text{max}\{\iota_1(s),\iota_2(s) \}$ with
	\begin{equation*}\label{ineq:smallgain psi2}
		\begin{aligned}
			&\iota_1(s)=(c_3\alpha_{\vartheta_2}\circ\chi_0)\circ (\alpha_{z_1}\circ \bar\alpha_{z_1}^{-1})^{-1}\circ (c_4\tau_{z_1}\circ\underline{\alpha}_{z_1}^{-1})(s)\\
			&\iota_2(s)=(c_3\bar\varphi_{\vartheta_2}\circ\chi_0)\circ\chi_4^{-1}\circ(c_4\tau_{\eta_2}/\lambda_{min}(Q_{\eta_2}))(s)
		\end{aligned}
	\end{equation*}
	for some $c_3, c_4>1$ with $\chi_0(s)=\sqrt{2s}$, $\chi_4(s)=(g_2\mu_0-2|H_2|^2| P_{\eta_2}|^2-2L_{\phi_3}|P_{\eta_2}|)s/{\lambda_{max}(P_{\eta_2})}$.
	Then by small gain theorem \cite{jiang1994small}, we can conclude that the  overall system composed of  \eqref{nonlinear plant}, \eqref{eq:xi}, \eqref{internal model 1}, and \eqref{internal model 2}  is globally asymptotically stable convergent to the steady-state manifold ${y}_r=\pi_{1}(w)$. The result of the theorem immediately follows.
\end{proof}

\section{Simulation}\label{sec:simulation}

We consider the nonlinear plant (see \cite{huang2004nonlinear}) described by
\begin{equation*}\label{eq:sim_plant}
	\begin{aligned}
		\dot{x}_0 &= -5 x_0^3 +  x_0^2 x_1, \\
        \dot{x}_1 &= x_0 + 0.1x_1 + 12.5 x_2, \\
        \dot{x}_2 &= -0.2 w_1 + \left(0.1 x_0 - 0.8 w_2\right)x_1 + \sin^2(w_1 x_2) + u, \\
        y_r &= x_1 ,
	\end{aligned}
\end{equation*}
driven by the neutrally-stable exo-system $\dot w_1=-0.5w_2, \dot w_2=0.5w_1$.
The cost is chosen as
\begin{equation*}\label{eq:sim_cost}
	F_c(x_1,w)=w_1^6-2x_1w_1^4+\frac{1}{2}w_1^4+x_1^2w_1^2-x_1w_1^2+\frac{1}{2}x_1^2.
 \end{equation*}
whose minimizer is $x_1^*=\pi_1(w)=w_1^2$.

We apply the proposed two-layer optimal regulation architecture. The optimization layer is a prime-form internal model driven by gradient measurements with $\kappa\equiv0$,
\begin{equation*}\label{eq:sim_opt_layer}
	\dot\xi_r=A_r\xi_r+B_r\phi_1(\xi_r)-G_r\,\nabla F_c(y_r,w),\qquad y^r_{r}=C_r\xi_r,
\end{equation*}
where $(A_r,B_r,C_r)$ is the prime triple, and
$G_r=\mathrm{col}(g_r\lambda_1,g_r^2\lambda_2,g_r^3\lambda_3)$ with $(\lambda_1,\lambda_2,\lambda_3)=(3,3,1)$.

The control layer consists of two pre-processing internal models and stabilizers
\begin{equation*}\label{eq:sim_ctrl_layer}
	\begin{aligned}
		& \dot\eta_i=A_i\eta_i+B_i\phi_{i+1}(\eta_i)+G_i v_i,\quad
		v_i=-\rho_i(\vartheta_i),\\
		&u=C_2\eta_2+D_2v_2,
	\end{aligned}
\end{equation*}
where $(A_i,B_i,C_i)$ ($i=1,2$) are prime triples of order $3$,
and
\(
	G_1=\mathrm{col}(g_1\nu_1,g_1^2\nu_2,g_1^3\nu_3),\quad 
	G_2=\mathrm{col}(g_2\mu_1,g_2^2\mu_2,g_2^3\mu_3),
\)
with $(\nu_1,\nu_2,\nu_3)=(\mu_1,\mu_2,\mu_3)=(3,3,1)$.
The error coordinates are
\(
\vartheta_1=x_1-y_{r}^r,\qquad
	\vartheta_2=x_2-C_1\eta_1-D_1v_1,
\)
and the stabilizers are selected as
\(
v_1=-15\,\vartheta_1(1+\vartheta_1^2)^{4},\qquad
	v_2=-40\,\vartheta_2(1+\vartheta_2^2)^{3},
\)
and we set
\(
	g_r=1.5,\quad g_1=4.0,\quad g_2=3.5,\quad D_1=D_2=1.
\)

Fig.~\ref{fig:error} shows the regulation errors $y_r^r-\pi_1(w)$, $x_1-\pi_1(w)$, and the gradients $\partial F_c/\partial x_1$ evaluated along the closed-loop trajectory.
The figure shows that the generated reference and the regulated output both converge to the optimal steady-state trajectory, while the measured gradient signal vanishes asymptotically.

\begin{figure}[ht]
	\centering
	\includegraphics[width=1.0\linewidth]{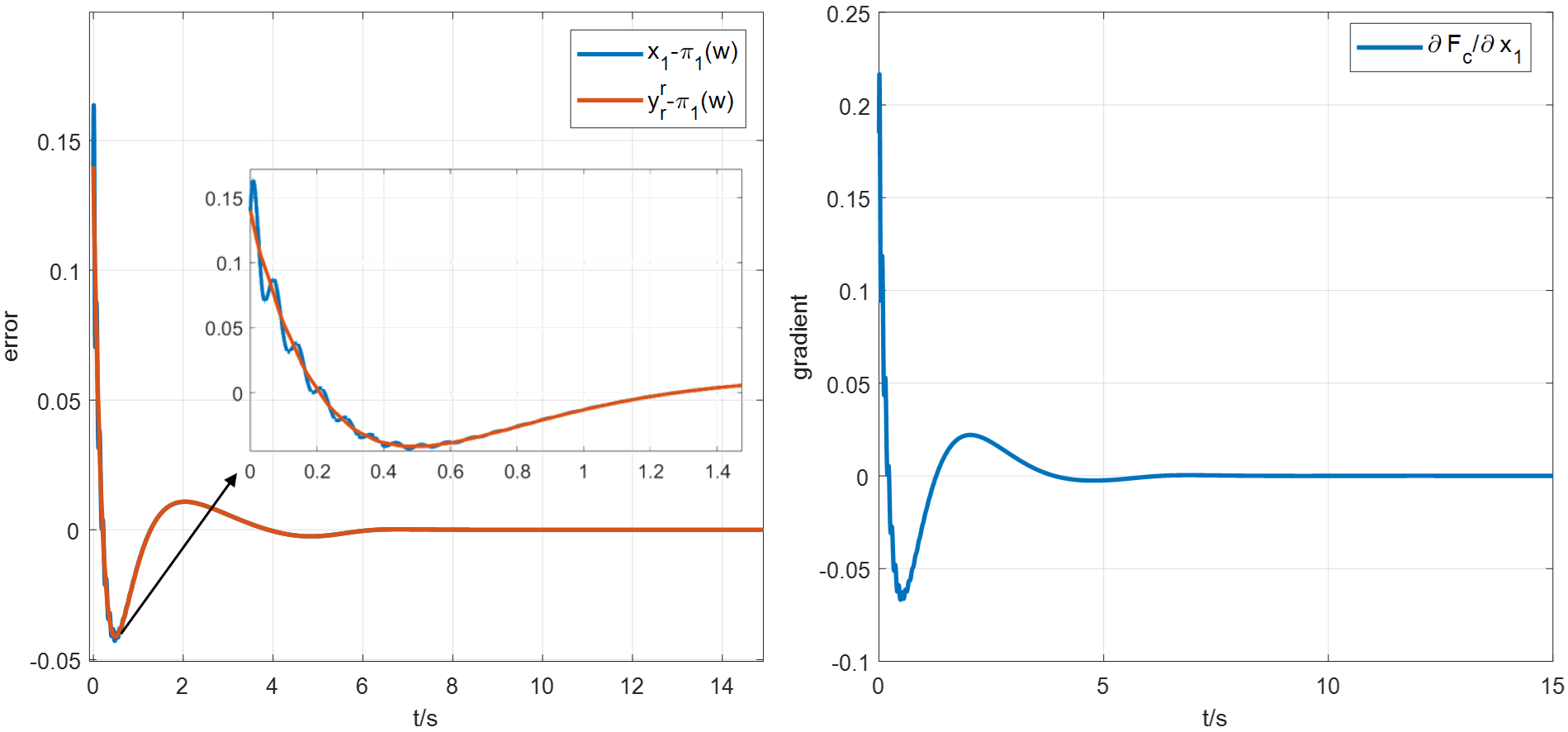}
	\caption{Trajectories of the regulation errors and real-time gradients.}
	\label{fig:error}
\end{figure}


\section{Conclusion}\label{sec:conclusion}
In this paper, we studied the OSSR problem for a class of nonlinear systems in lower-triangular form subject to disturbances generated by a neutrally stable nonlinear exo-system. By introducing the GRE, we incorporated the steady-state optimality condition into a constructive nonlinear regulation framework. On this basis, we proposed a two-layer structure in which a steady-state generator combined with gradient-feedback dynamics serves as the optimization layer to generate the optimal reference, while a backstepping-based state-feedback controller in the control layer ensures tracking. Future research will focus on extending the framework to handle state and input constraints, output-feedback designs with unknown exo-systems, and applications to large-scale systems.

\section{Appendix}

\subsection{Proof of Lemma \ref{lemma 1}}\label{prof of lemma1}
\begin{proof}
	Letting $D_{g_r}=\text{diag}(g_r,g_r^2,...,g_r^{d_1})$ and $\bar{\xi}_r=D_{g_r}^{-1}\tilde{\xi}_r$,
	we obtain that
	\begin{equation} \label{eq:bar xi}
		\begin{aligned}
			\dot{\bar{\xi}}_r=&g_r(A_r-\omega(\bar y_r^r,w)G_0C_r)\bar{\xi}_r+B_r\Delta(\bar{\xi}_r,\pi_{\xi_r}(w))\\
			&-G_0\Omega_2(e,w)+D_{g_r}^{-1}\kappa(e)\\
		\end{aligned}
	\end{equation}
	where $G_0=\text{col}(\lambda_1,...,\lambda_{d_1})$, $\Delta(\bar{\xi}_r,\pi_{\xi_r}(w))=\dfrac{1}{g_r^{d_1}}(\phi_{1}(D_{g_r}\bar{\xi}_r+\pi_{\xi_r}(w))-\phi_{1}(\pi_{\xi_r}(w)))$.
	
	It is further observed that the function $\phi_{1}(\cdot)$ is globally Lipschitz by Assumption \ref{immersion condition assumption}, and we have
	\begin{equation*} 
		\begin{aligned}
			|\Delta(\bar{\xi}_r,\pi_{\xi_r}(w))|\leq\dfrac{L_{\phi_1}}{g_r^{d_1}}|D_{g_r}\bar{\xi}_r|\leq L_{\phi_1}|\bar{\xi}_r|
		\end{aligned}
	\end{equation*}
	for any $g_r>1$,  where $L_{\phi_1}$ is the Lipschitz constant.
	
	By using Lemma \ref{Lemma Hurwitz}, there exists $G_0$ which makes the matrix $M_r(t)=A_r-\omega(w(t))G_0C_r$ Hurwitz,
	where $\omega(\bar y_r^r,w)$ is denoted as $\omega(t)$ for simplicity, and $M_r(t)$ is in the form as $M_n(t)$ in Lemma \ref{Lemma Hurwitz}, 
	which is to say there exists $\lambda_0>0$ and positive definite matrix $P$ such that
	\begin{equation*} 
		\begin{aligned}
			PM_r(t)+M_r^\top(t)P\leq -\lambda_0I 
		\end{aligned}
	\end{equation*}
	

	With these in mind, we construct function $V_r(\bar{\xi}_r)=\bar{\xi}_r^\top P \bar{\xi}_r$. Along the trajectory decided by \eqref{eq:bar xi}, we have
	\begin{equation*} \label{Vr dynamics}
		\begin{aligned}
			&\dot{V}_r(\bar{\xi}_r)\\&
			 \leq g_r\bar{\xi}_r^\top[M_r(t)^\top P+PM_r(t)]\bar{\xi}_r+(2+2L_{\phi_1}|P|)|\bar{\xi}_r|^2 \\
			 &+|PG_0|^2\Omega_2^2(e,w)+|PD_{g_r}^{-1}|^2|\kappa(e)|^2\\
			& \leq -(g_r\lambda_0-2-2L_{\phi_1}|P|)|\bar{\xi}_r|^2+\alpha_{e}(|e|)
		\end{aligned}
	\end{equation*}
	where $\alpha_{e}(\cdot)$ is a $\mathcal{K}_\infty$ function which satisfies $\alpha_{e}(|e|)\geq |PG_0|^2\Omega_2^2(e,w)+|PD_{g_r}^{-1}|^2|\kappa(e)|^2$.
	Thus there exists $g_r^*$ satisfying \(g_r^*\geq ({2+2L_{\phi_1}|P|})/{\lambda_0}\)
	such that by choosing $g_r>g_r^*$, the $\bar{\xi}_r$ subsystem is ISS, with $e$ as input and $\bar{\xi}_r$ as state.
	Therefore, by constructing $V_o(\tilde{\xi}_r)=\tilde{\xi}_r^\top Q\tilde{\xi}_r$, where $Q=(D_{g_r}^{-1})^\top PD^{-1}_{g_r}$, we have
	\begin{equation}\label{eq:dot Vo}
		\begin{aligned}
			\dot{V}_o(\tilde{\xi}_r)
			&\leq -\alpha_{\xi}|\tilde{\xi}_r|^2 +\alpha_{e}(|e|) \\
		\end{aligned}
	\end{equation}
	where $\alpha_{\xi}=D_g^{-2}(g_r\lambda_0I-I-2L_{\phi_1}|P|)>0$.
	
\end{proof}

\subsection{Proof of Lemma \ref{lemma 2}}\label{prof of lemma2}
We can rewrite the $\vartheta_1$ dynamics as
\begin{equation}\label{eq:vartheta1}
	\dot{\vartheta}_1=\varpi_{\vartheta_1}(\tilde{x}_0,\tilde \xi_r,\bar \eta_1,\vartheta_1,\vartheta_2,w)+b_1(w)D_1v_1
\end{equation}
where the smooth function $\varpi_{\vartheta_1}$ satisfies $\varpi_{\vartheta_1}(0,0,0,0,0,w)=0$,
and there exist positive constant $\beta_{\vartheta_2}$, $\beta_{\eta_1}$ and class $\mathcal{K}_\infty$ functions $\beta_{x_0}(\cdot)$,  $\beta_{\xi_r}(\cdot)$, $\beta_{\vartheta_1}(\cdot)$ such that
\begin{equation*}\label{ineq:varpi}
	\begin{aligned}
		&\vartheta_1{\varpi_{\vartheta_1}}(\tilde{x}_0,\tilde \xi_r,\bar \eta_1,\vartheta_1,\vartheta_2,w)\\
		&\leq \beta_{x_0}(|\tilde{x}_0|)+\beta_{\xi_r}(|\tilde \xi_r|)+\beta_{\eta_1}|\bar \eta_1|^2+\beta_{\vartheta_1}(|\vartheta_1|)+\beta_{\vartheta_2}|\vartheta_2|^2
	\end{aligned}
\end{equation*}

For simplicity, we rewrite the $\bar \eta_1$-dynamics as
\begin{equation*}\label{eq:bareta1}
	\begin{aligned}
		\dot{\bar{\eta}}_1=&M_{\eta_1}\bar{\eta}_1+B_1\tilde{\phi}_2(\bar{\eta}_1+b_1^{-1}(w)G_1D_1^{-1}\vartheta_1,w)\\&+G_1\Delta_{\eta_1}(\tilde{x}_0,\tilde \xi_r,\vartheta_1,\vartheta_2,w)
	\end{aligned}
\end{equation*}
where $M_{\eta_1}=A_1-G_1D_1^{-1}C_1$, and the smooth function $\Delta_{\eta_1}$ satisfies $\Delta_{\eta_1}(0,0,0,0,w)=0$,
and there exist positive constant $L_{\phi_2}^{'}$, $\gamma_{\vartheta_2}$ and class $\mathcal{K}_\infty$ functions$\gamma_{\xi_r}(\cdot)$, $\gamma_{x_0}(\cdot)$,  $\gamma_{\vartheta_1}(\cdot)$ such that
\begin{equation*}
	\begin{aligned}
		&|\tilde{\phi}_2(\bar{\eta}_1+b_1^{-1}(w)G_1D_1^{-1}\vartheta_1,w)|\leq L_{\phi_2}|\bar \eta_1|+ |G_1|L_{\phi_2}^{'}|\vartheta_1|\\
		&|{\Delta_{\eta_1}}|^2\leq \gamma_{x_0}(|\tilde{x}_0|)+\gamma_{\xi_r}(|\tilde \xi_r|)+\gamma_{\vartheta_1}(|\vartheta_1|)+\gamma_{\vartheta_2}|\vartheta_2|\\
	\end{aligned}
\end{equation*}
and we let $\bar\gamma_{\vartheta_1}(|\vartheta_1|)=L_{\phi_2}^{'2}|\vartheta_1|^2+\gamma_{\vartheta_1}(|\vartheta_1|)$, which is a $\mathcal{K}_\infty$ function.

We consider the $(\tilde{\xi}_r,\tilde{x}_0,\bar{\eta}_1,\vartheta_1)$-subsystem, where $\vartheta_1$ is the external dynamics with $v_1$ as input and $\vartheta_1$ as output, whose dynamics are given in \eqref{eq:vartheta1}, and $(\tilde{\xi}_r,\tilde{x}_0,\bar{\eta}_1)$-subsystem is regarded as zero dynamics.
	With assumption \ref{strong minimum phase property}, from \eqref{ineq:minimum phase} we have
	\begin{equation}\label{ineq:Vtildex0}
		\begin{aligned}
			\dot{V}_0(\tilde{x}_0)\leq-\alpha_0(|\tilde{x}_0|)+\alpha_{\vartheta_1}(|\vartheta_1|)+\alpha_{\xi_r}(|\tilde \xi_r|)
		\end{aligned}
	\end{equation}
	where $\alpha_{\vartheta_1}(\cdot)$ and $\alpha_{\xi_r}(\cdot)$ are bounded $C^1$ class $\mathcal{K}_\infty$ functions.
    Then we study the stability of the interconnected systems $(\tilde\xi_r,\vartheta_1)$, $(\bar\eta_1,\vartheta_1)$ and $(\tilde x_0,\vartheta_1)$.
	
	For the $\vartheta_1$-subsystem, we construct the Lyapunov function $V_{\vartheta_1}=1/2\vartheta_1^2$. 
	For $\bar{\eta}_1$-subsystem, we define $\hat{\eta}_1=D_{g_1}^{-1}\bar{\eta}_1$, where $D_{g_1}=\text{diag}(g_1,...,g_1^{d_2})$, and we have 
	\begin{equation} \label{eq:hat eta}
		\begin{aligned}
			\dot{\hat{\eta}}_1=&g_1(A_1-H_1D_1^{-1}C_1)\hat{\eta}_1+g_1^{-d_2}L_{\phi_2}\tilde\phi_2+H_1\Delta_{\eta_1}
		\end{aligned}
	\end{equation}
	Define $M_{\eta_1}^{'}=A_1-H_1D_1^{-1}C_1$, and we construct $V_{\eta_1}(\hat{\eta}_1)=\hat{\eta}_1^\top P_{\eta_1}\hat{\eta}_1$, where $P_{\eta_1}$ is a positive definite matrix satisfying $P_{\eta_1}M_{\eta_1}^{'}+M^{'\top}_{\eta_1}P_{\eta_1}\leq -\nu_0I$
	for some $\nu_0>0$.
	Along the trajectory decided by \eqref{eq:hat eta}, we have
    
	\begin{equation*}
		\begin{aligned}
			&\dot{V}_{\eta_1}(\hat{\eta}_1)
			\leq g_1\hat{\eta}_1^\top(P_{\eta_1}M_{\eta_1}^{'}+M^{'\top}_{\eta_1}P_{\eta_1})\hat{\eta}_1+2(H_1\Delta_{\eta_1})^\top P_{\eta_1}\hat{\eta}_1 \\
			&+2(g_1^{-d_2}\tilde\phi_2)^\top  P_{\eta_1}\hat{\eta}_1\\
			&\leq -(g_1\nu_0-2|H_1|^2| P_{\eta_1}|^2-2L_{\phi_2}|P_{\eta_1}|)|\hat{\eta}_1|^2\\
			& +\gamma_{x_0}(|\tilde{x}_0|)+\gamma_{\xi_r}|\tilde \xi_r|+\bar\gamma_{\vartheta_1}(|\vartheta_1|)+\gamma_{\vartheta_2}|\vartheta_2|
		\end{aligned}
	\end{equation*}
	By reconstructing  $V_{\bar\eta_1}(\bar{\eta}_1)=\bar{\eta}_1^\top Q_{\eta_1}\bar{\eta}_1$, where $Q_{\eta_1}=(D_{g_1}^{-1})^\top P_{\eta_1}D_{g_1}^{-1}$, we have
	\begin{equation*}
		\begin{aligned}
			&\dot{V}_{\bar \eta_1}(\bar{\eta}_1)
			\leq -(g_1\nu_0-2|H_1|^2|P_{\eta_1}|^2-2L_{\phi_2}|P_{\eta_1}|)|D_{g_1}^{-1}\bar{\eta}_1|^2\\
			& +\gamma_{x_0}(|\tilde{x}_0|)+\gamma_{\xi_r}|\tilde \xi_r|+\bar\gamma_{\vartheta_1}(|\vartheta_1|)+\gamma_{\vartheta_2}|\vartheta_2|\\
			&\leq -(g_1\nu_0-2|H_1|^2| P_{\eta_1}|^2-2L_{\phi_2}|P_{\eta_1}|)V_{\bar \eta_1}(\bar{\eta}_1)/{\lambda_{max}(P_{\eta_1})}\\
			& +\gamma_{x_0}(|\tilde{x}_0|)+\gamma_{\xi_r}|\tilde \xi_r|+\bar\gamma_{\vartheta_1}(\sqrt{2V_{\vartheta_1}(\vartheta_1)})+\gamma_{\vartheta_2}|\vartheta_2|
		\end{aligned}
	\end{equation*}
	Thus, we can prove that there exists a $g_1^*$ which satisfies 
	\[
	g_1^*\geq 2(|H_1|^2| P_{\eta_1}|^2+L_{\phi_2}|P_{\eta_1}|)/\nu_0
	\]
	such that for any $g_1>g_1^*$, the $\bar{\eta}_1$ dynamics is ISS with respect to the input $(\tilde \xi_r,\tilde x_0,\vartheta_1,\vartheta_2)$. 
	
	Along the trajectory decided by \eqref{eq:vartheta1}, with the feedback control law designed as $v_1=-\rho_1(\vartheta_1)$, we have
	\begin{equation*}
		\begin{aligned}
			&\dot{V}_{\vartheta_1}(\vartheta_1)=\vartheta_1(\varpi_{\vartheta_1}(\tilde{x}_0,\tilde \xi_r,\bar \eta_1,\vartheta_1,\vartheta_2,w)-b_1(w)D_1 \rho_1(\vartheta_1))\\
			&\leq  -(b_{min}D_1 \rho_1(\vartheta_1)\vartheta_1-\beta_{\vartheta_1}(|\vartheta_1|))+ \beta_{x_0}(\underline{\alpha}_0^{-1}(V_0(\tilde x_0)))\\
			&+\beta_{\xi_r}(\sqrt{\dfrac{V_o(\tilde\xi_r)}{\lambda_{min}(Q)}})
			+\beta_{\eta_1}{\dfrac{V_{\bar\eta_1}(\bar\eta_1)}{\lambda_{min}(Q_{\eta_1})}}+\beta_{\vartheta_2}|\vartheta_2|^2\\
		\end{aligned}
	\end{equation*}
	With \eqref{ineq:Vtildex0}, we have
	\[
	\dot{V}_0(\tilde{x}_0)\leq-\alpha_0(\bar\alpha_0^{-1}(V_0(\tilde{x}_0)))+\alpha_{\vartheta_1}(\sqrt{2V_{\vartheta_1}(\vartheta_1)})+\alpha_{\xi_r}(|\tilde \xi_r|)
	\]
	and with  \eqref{eq:dot Vo}, we have
	\[
	\dot{V}_o(\tilde{\xi}_r)
	\leq -{\alpha_{\xi}}V_o(\tilde\xi_r)/{\lambda_{max}(Q)} +\alpha_{e}(\sqrt{2V_{\vartheta_1}(\vartheta_1)}) 
	\]
	To guarantee the stability of the interconnected systems, with any $g_1\geq g_1^*$, we design the smooth function $\rho_1(\cdot)$ to satisfy the following
	\begin{equation*}\label{ineq:smallgain rho1}
		\begin{aligned}
			\rho_1(\vartheta_1)\geq (\beta_{\vartheta_1}(\vartheta_1)+\psi_1(1/2\vartheta_1^2))/({b_{min}|D_1|\vartheta_1})
		\end{aligned}
	\end{equation*}
	where $\psi_1(\cdot)>\text{max}\{\varrho_1(s),\varrho_2(s),\varrho_3(s) \}$ with
	\begin{equation*}\label{ineq:smallgain psi1}
		\begin{aligned}
			&\varrho_1(s)=(c_1\alpha_{\vartheta_1}\circ\chi_0)\circ (\alpha_0\circ \bar\alpha_0^{-1})^{-1}\circ (c_2\beta_{x_0}\circ\underline{\alpha}_0^{-1})(s)\\
			& \varrho_2(s)=(c_1\alpha_{e}\circ\chi_0)\circ (c_2\lambda_{max}(Q)/{\alpha_\xi})\beta_{\xi_r}\circ \chi_1(s)\\
			&\varrho_3(s)=(c_1\bar\gamma_{\vartheta_1}\circ\chi_0)\circ\chi_3^{-1}\circ(c_2\beta_{\eta_1}/\lambda_{min}(Q_{\eta_1}))(s)
		\end{aligned}
	\end{equation*}
	for some $c_1, c_2>1$ with $\chi_0(s)=\sqrt{2s}$, $\chi_1(s)=\sqrt{s/\lambda_{min}(Q)}$, $\chi_3(s)=\dfrac{s}{\lambda_{max}(P_{\eta_1})}(g_1\nu_0-2|H_1|^2| P_{\eta_1}|^2-2L_{\phi_2}|P_{\eta_1}|)$.  Then by small gain theorem \cite{jiang1994small}, we can conclude that the $(\tilde{z}_1,\vartheta_1)$-subsystem is ISS, with input $\vartheta_2$, which completes the proof.
\bibliographystyle{IEEEtran} 
\bibliography{ref}             
\end{document}